\documentclass[a4paper,11pt,reqno]{amsart}
\usepackage[colorlinks=true,linkcolor=blue,citecolor=blue,plainpages=false,pdfpagelabels]{hyperref}
\usepackage{amsmath,amssymb,amsthm,amsfonts}
\usepackage{mathtools}
\usepackage{nicefrac} 
\usepackage[mathscr]{euscript} 
\usepackage{microtype}
\usepackage{enumerate}

\usepackage{bm}
\newcommand*{\mathbold}[1]{{\bm #1}}

\DeclareMathOperator{\tr}{tr}

\newcommand{\cH}{\mathscr{H}}

\newcommand*{\ci}{\mathrm{i}}
\newcommand*{\di}{\mathrm{d}} 
\newcommand{\eps}{\varepsilon}
\newcommand{\cP}{\mathscr{P}(\cH)}
\newcommand{\cPp}{\mathscr{P}_+(\cH)}
\newcommand{\cL}{\mathscr{L}}
\newcommand{\bR}{\mathbb{R}}
\newcommand{\bN}{\mathbb{N}}

\newcommand{\unorm}[1]{|\!|\!| #1 |\!|\!|}
\newcommand{\unormb}[1]{\Big|\!\Big|\!\Big| #1 \Big|\!\Big|\!\Big|}

\renewcommand{\vec}[1]{\mathbold{#1}}

\renewcommand{\a}{\vec{a}}
\renewcommand{\b}{\vec{b}}
\renewcommand{\c}{\vec{c}}

\newtheorem{lemma}{Lemma}
\newtheorem{proposition}[lemma]{Proposition}
\newtheorem{theorem}[lemma]{Theorem}
\newtheorem{corollary}[lemma]{Corollary}
\theoremstyle{plain}

\newtheorem{remark}{Remark}

\begin{document}
\nocite{*}
\title[Generalized Log-Majorization and Multivariate Trace Inequalities]{Generalized Log-Majorization\\
and Multi\-variate Trace Inequalities}

\author{Fumio Hiai}
\address{Tohoku University (Emeritus)}
\email{hiai.fumio@gmail.com}

\author{Robert K\"onig}
\address{Institute for Advanced Study \& \\
Zentrum Mathematik\\
Technische Universit\"at M\"unchen\\
Garching, Germany}
\email{robert.koenig@tum.de} 

\author{Marco Tomamichel}
\address{School of Physics\\
The University of Sydney\\
Sydney, Australia}
\email{marcotom.ch@gmail.com}

\begin{abstract}
We show that recent multivariate generalizations of the Araki-Lieb-Thirring inequality and the Golden-Thompson inequality [Sutter, Berta, and To\-ma\-michel, {\em Comm.  Math. Phys.} (2017)] for Schatten norms hold more generally for all unitarily invariant norms and certain variations thereof. The main technical contribution is a generalization of the concept of log-majorization which allows us to treat majorization with regards to logarithmic integral averages of vectors of singular values.
\end{abstract}
\maketitle

\section{Introduction}
Majorization and log-majorization are powerful and versatile tools for proving trace and norm inequalities (see, e.g.,~\cite{ando89,marshall11,hiaipetz14} for overviews on the topic). A fundamental property of unitarily invariant norms (including Schatten $p$-norms and the trace norm) can roughly be stated as follows (see~\cite[Thm.~IV.2.2]{bhatia97} or~\cite[Prop.~4.4.13]{hiai10b}):
\begin{flushleft}
\hspace{1cm}\parbox{10.2cm}{\em For two matrices $A$ and $B$, the singular values of $A$ are {weakly} majorized by the singular values of $B$ if and only if $\unorm{A} \leq \unorm{B}$ for every unitarily invariant norm $\unorm{\cdot}$.}\hspace*{0.82cm}($\star$)\\[0.2cm]
\end{flushleft}
A natural approach to prove norm inequalities for general unitarily invariant norms  then proceeds as follows: First, the desired inequality is shown for the operator norm where such inequalities often boil down to operator inequalities and are easier to prove. Next, it is shown using antisymmetric tensor power calculus that the operator norm inequality implies log-majorization and thus weak majorization of the eigenvalues. Consequently, the desired inequalities  follow directly from ($\star$).

Let us illustrate this approach with an example (the reader unfamiliar with the notation is referred to Section~\ref{sec:not}). For two positive definite operators $A_1, A_2$ and any $\theta \in (0,1)$,
consider the following special case of the Araki-Lieb-Thirring inequality~\cite{lieb76,araki90}:
\begin{align}
  \left\| \left( A_1^{\frac{\theta}{2}} A_2^{\theta} A_1^{\frac{\theta}{2}}\right)^{\frac{1}{\theta}} \right\| \leq \left\| A_1^{\frac12} A_2 A_1^{\frac12}\right\| , \label{eq:lt}
\end{align} 
where $\|\cdot\|$ denotes the operator norm. Its elementary proof, using operator monotonicity of $t \mapsto t^{\theta}$ as its main ingredient, simply argues that
\begin{align}
  A_1^{\frac12} A_2 A_1^{\frac12} \leq I 
  \implies A_2 \leq A_1^{-1}
  \implies A_2^{\theta} \leq A_1^{-\theta}
  \implies A_1^{\frac\theta2} A_2^{\theta} A_1^{\frac\theta2} \leq I  \,.
\end{align}
Inequality~\eqref{eq:lt} for positive semi-definite operators follows by continuity.
Applying~\eqref{eq:lt} to the antisymmetric powers $\land^k A_1$ and $\land^k A_2$, using Properties~\eqref{it:asymprod}, \eqref{it:asympow} and~\eqref{it:asymeig} of the antisymmetric tensor power discussed in the next section, we find
\begin{align}
  \vec{\lambda}\left( \left( A_1^{\frac{\theta}{2}} A_2^{\theta} A_1^{\frac{\theta}{2}}\right)^{\frac{1}{\theta}} \right) \prec_{\log} \vec{\lambda}\left( A_1^{\frac12} A_2 A_1^{\frac12} \right) .
\end{align}
Here $\vec{\lambda}(A)$ is a vector comprising the eigenvalues of $A$ in decreasing order counting multiplicities and $\prec_{\log}$ denotes log-majorization. 
Since log-ma\-jo\-rization implies {weak} majorization, the relation~($\star$) allows us to lift~\eqref{eq:lt} to arbitrary unitarily invariant norms, including the trace. In fact, log-majorization is stronger than {weak} majorization and thus allows us to derive stronger norm inequalities (see, e.g.,~\cite{hiai10b}).

{The use of the antisymmetric tensor power approach has so far been  restricted to matrix functions made from operations of products, absolute values and powers (see, e.g.,~\cite{araki90} (explained above) and~\cite{AH94,ando94}). In this work we extend the approach to settings with a logarithmic integral average} so that it can be applied to recent multivariate trace inequalities~\cite{sutter16}. For example,~\cite[Thm.~3.2]{sutter16} specialized to the operator norm and three positive semi-definite matrices $A_1, A_2, A_3$ and $\theta \in (0, 1)$ generalizes the Araki-Lieb-Thirring inequality and reads
\begin{align}
    \log \left\| \left| A_1^{\theta} A_2^{\theta} A_3^{\theta} \right|^{\frac{1}{\theta}} \right\|
    \leq \int_{-\infty}^{\infty}  \log \left\| A_1 A_2^{1+\ci t} A_3 \right\|\, \di\beta_{\theta}(t) \,, \label{eq:three1}
\end{align}
where $\di \beta_\theta(t)$ is some probability measure on $\mathbb{R}$. Using the antisymmetric tensor power technique this can be transformed into a log-majorization relation, namely\footnote{The details of this derivation are given in Section~\ref{sec:app}.}
\begin{align} 
   \vec{\lambda}\left( \left| A_1^{\theta} A_2^{\theta} A_3^{\theta} \right|^{\frac{1}{\theta}} \right) \prec_{\log} 
  \exp \int_{-\infty}^{\infty} \log \vec{\lambda} \left( \left| A_1 A_2^{1+\ci t} A_3 \right|  \right)\, \di\beta_{\theta}(t) \,. \label{eq:three2}
\end{align}
However, known results for majorization or log-majorization in the spirit of ($\star$) do not apply to~\eqref{eq:three2} due to the integral average of vectors on the right-hand side.

In Sections~\ref{sec:mainweak} we extend ($\star$) to the case of weak majorization relations where the right-hand side contains an integral average of vectors.
Our first main result, Theorem~\ref{th:mainweak} in Section~\ref{sec:main}, 
 deals with weak log-majorization relations of the form~\eqref{eq:three2}. It establishes that the weak log-majorization relation is equivalent to two other conditions involving unitarily invariant norms, and in particular implies that~\eqref{eq:three1} holds for all unitarily invariant norms and certain variations thereof. 
Our second main result, split into Theorems~\ref{th:main} and~\ref{th:mainadditional} in Section~\ref{sec:logmajorizationintegral}, proves a similar characterization directly for the log-majorization relation in~\eqref{eq:three2} and implies even stronger inequalities for unitarily invariant norms. For the special case where no average is present, Propositions~\ref{pr:charweaklog} and~\ref{pr:charlog} imply new characterizations of weak log-majorization and log-majorization, respectively.
 The implications for multivariate trace inequalities are discussed in Section~\ref{sec:app}. There we present multivariate generalizations of the Araki-Lieb-Thirring inequality,  the Golden-Thompson inequality~\cite{golden65,thompson65} and Lieb's triple matrix inequality~\cite{lieb73a}\,---\,beyond the generalizations recently established in~\cite{sutter16}. We also provide a simplified proof of~\eqref{eq:three1} for arbitrarily many matrices and the operator norm in Appendix~\ref{app:shorter}.

\section{Preliminaries}
\label{sec:not}

\subsection*{Majorization} Let $\cH$ be a Hilbert space of dimension $d := \dim\cH < \infty$, $\cL(\cH)$ the set of linear operators on $\cH$, $\cP$ be the
set of all positive semi-definite operators in $\cL(\cH)$, and $\cPp$ the set of all invertible (positive definite)
operators in $\cP$. For {self-adjoint} $A, B \in \cL(\cH)$, we write $A \geq B$ to indicate that $A - B \in \cP$.

We use bold font $\a = (a_1,\dots,a_d) \in \bR^d$ to denote vectors. Let $\bR^d_+ := \{ \a \in\bR^d : a_1,\dots,a_d \ge 0 \}$. For $\a$,
$\b \in{\bR^d}$ such that $a_1\ge\dots\ge a_d$ and $b_1\ge\dots\ge b_d$,
\emph{weak majorization}, denoted {by} $\a\prec_w\b$, is the relation
\begin{align}
\sum_{i=1}^ka_i\le\sum_{i=1}^kb_i,\qquad k \in [d] \,,
\end{align}
where we used $[d] := \{1, 2, \ldots, d\}$.
\emph{Majorization}, $\a\prec\b$, additionally requires that equality holds for $k=d$.
 {For $\a,\b\in\bR^d_+$ such that  $a_1\ge\dots\ge a_d$ and $b_1\ge\dots\ge b_d$,}
\emph{weak log-majorization}, $\a\prec_{w\log}\b$, is the relation
\begin{align}
\prod_{i=1}^ka_i\le\prod_{i=1}^kb_i,\qquad k \in [d] ,
\end{align}
and \emph{log-majorization}, $\a \prec_{\log} \b$, additionally requires equality for $k = d$.
For any function $f$ on $\bR_+$ we write $f(\a) = (f(a_1), \dots, f(a_d))$
with conventions $\log 0:=-\infty$ and $e^{-\infty}=0$. Moreover, weak majorization $\prec_w$ makes sense even for vectors having entries $-\infty$. With these conventions, it is evident that $\a \prec_{w\log} \b$ if and only if $\log \a \prec_{w} \log \b$. The following relation takes a prominent role (see~\cite[Thm.~II.3.3]{bhatia97} and \cite[Prop.~4.1.4]{hiai10b}): 
\begin{lemma}
\label{lm:convex}
 For any convex function $f: [0,\infty)\to[0,\infty)$, we have $\a \prec \b \implies f(\a) \prec_{w} f(\b)$. Moreover, if $f$ is also non-decreasing, then $\a \prec_w \b \implies f(\a) \prec_{w} f(\b)$.
\end{lemma}
As a direct consequence when applied to the exponential function, weak log-majorization implies weak majorization. 

\subsection*{Unitarily invariant norms} Let us denote the eigenvalues of a self-adjoint operator~$A \in \cL(\cH)$ in decreasing order counting multiplicities by the vector $\vec{\lambda}(A)=(\lambda_1(A),\dots,\lambda_d(A))$.
Moreover, let $\unorm{\cdot}_{\Phi}$ be a unitarily invariant norm on $\cP$ and $\Phi: \bR^d_+ \to \bR_+$ the
corresponding gauge function so that
\begin{align}
\unorm{L}_{\Phi} = \unorm{|L|}_{\Phi} = \Phi(\vec{\lambda}(|L|)) \,.
\end{align} 
(We refer the reader to~\cite[Sec.~IV]{bhatia97} for an introduction to unitarily invariant norms. The bijective correspondence between symmetric gauge functions on~$\mathbb{R}^d_+$ and unitarily invariant norms on~$\cP$ is due to von Neumann~\cite{vonNeumann77}.) Of particular interest here are Ky Fan norms. For $k \in [d]$, the Ky Fan $k$-norm, $\| \cdot \|_{(k)} : \cL(\cH) \to \bR_+$, is defined as
\begin{align}
   L \mapsto \| L \|_{(k)} := \sum_{i=1}^k \lambda_i(|L|) \,.
\end{align}
In particular, $\|\cdot\|_{(1)}$ is the operator norm $\|\cdot\|$. Another important and familiar one is the Schatten $p$-norm $\|L\|_p:=(\tr |L|^p)^{1/p}$ for $p\geq 1$. In particular, $\|\cdot \|_1$ is the trace norm. The definition of $\|\cdot\|_p$ makes sense even for $0<p<1$ as a quasi-norm. 

The following lemma is a H\"older inequality for the gauge
function $\Phi$ and follows from~\cite[Thm.~IV.1.6]{bhatia97}.

\begin{lemma}\label{lm:hoelder}
For $l \in [m]$ let $\a_l=(a_{l1},\dots,a_{ld})\in\bR^d_+$ and $\beta_l>0$ with
$\sum_{l=1}^m\beta_l=1$. Then
\begin{align}
\Phi\Biggl(\prod_{l=1}^m\a_l^{\beta_l}\Biggr)\le\prod_{l=1}^m\Phi^{\beta_l}(\a_l),
\end{align}
where
\begin{align}
\prod_{l=1}^m\a_l^{\beta_l}:=\Biggl(\prod_{l=1}^ma_{l1}^{\beta_l},\dots,
\prod_{l=1}^ma_{ld}^{\beta_l}\Biggr).
\end{align}
\end{lemma}

\begin{proof}
The lemma for $m=2$ is~\cite[Thm.~IV.1.6]{bhatia97}. The case $m=3$ is shown as
\begin{align}
\Phi\bigl(\a_1^{\beta_1}\a_2^{\beta_2}\a_3^{\beta_3}\bigr)
&=\Phi\biggl(\Bigl(\a_1^{\frac{\beta_1}{\beta_1+\beta_2}}\a_2^{\frac{\beta_2}{\beta_1+\beta_2}}
\Bigr)^{\beta_1+\beta_2}\a_3^{\beta_3}\biggr) \\
&\le\Phi^{\beta_1+\beta_2}\Bigl(\a_1^{\frac{\beta_1}{\beta_1+\beta_2}}
\a_2^{\frac{\beta_2}{\beta_1+\beta_2}}\Bigr)\Phi^{\beta_3}(\a_3)
\le\Phi^{\beta_1}(\a_1)\Phi^{\beta_2}(\a_2)\Phi^{\beta_3}(\a_3).
\end{align}
The general case can be shown similarly by induction.
\end{proof}

\subsection*{Antisymmetric tensor product}
For $k \in [d]$, let $\cH^{\otimes k}$ be the $k$th tensor power of $\cH$ and let $\cH^{\land k}$ denote the antisymmetric subspace of $\cH^{\otimes k}$.
The $k$th \emph{antisymmetric tensor power}, $\land^k: \cL(\cH) \to \cL(\cH^{\land k})$, 
maps any linear operator $L$ to the restriction of $L^{\otimes k} \in \cL(\cH^{\otimes k})$ to the antisymmetric subspace $\cH^{\land k}$ of $\cH^{\otimes k}$. It satisfies the following rules (see, e.g.,~\cite[Sec.~I.5 and~p.~18]{bhatia97}):
\begin{lemma} 
\label{lm:anti}
 Let $L,K \in \cL(\cH)$ and $A \in \cP$. For any $k \in [d]$, we have:
\begin{enumerate}[(a)]
  \item $(\land^k L)^{\dag} = \land^k (L^{\dag})$,\label{it:asymfirst}
  \item $(\land^k L) (\land^k K) = \land^k (LK)$,\label{it:asymprod}
  \item $(\land^k A)^{z} = \land^k (A^z)$ for all $z \in \mathbb{C}$, and\label{it:asympow}
  \item $\big\| \land^{k} L \big\| = \prod_{i=1}^{k} \lambda_i(|L|)$.\label{it:asymeig} 
\end{enumerate}
\end{lemma}
In particular, we note that if $L\in\cL(\cH)$ is positive semi-definite, so is its antisymmetric tensor power~{$\land^k L \in \mathcal{L}(\cH^{\land k})$}.

\section{(Weak) majorization with integral average}
\label{sec:mainweak}

 Let $\Xi$ be a $\sigma$-compact metric space and $\nu$ a probability measure on the Borel
$\sigma$-field of $\Xi$. Let $A\in\cL(\cH)$ and $\xi\in\Xi\mapsto B_\xi\in\cL(\cH)$ be a continuous function such that $A$ and $B_\xi$ for all $\xi\in\Xi$ are self-adjoint and $\sup\bigl\{\|B_\xi\| : \xi\in\Xi\bigr\} < \infty$. We use the convention
\begin{align}
\int_\Xi \vec{\lambda}(B_\xi)\,\di\nu(\xi)
&:=\biggl(\int_\Xi\lambda_1(B_\xi)\,\di\nu(\xi),\dots,\int_\Xi\lambda_d(B_\xi)\,\di\nu(\xi)\biggr) . \label{eq:conventionintegralvec}
\end{align}
The following two theorems are characterizations of weak majorization and majorization in the setting with integral average. They will be used in Sections~\ref{sec:main} and~\ref{sec:logmajorizationintegral}. 

\begin{theorem}\label{th:mainweaksimple}
With $\Xi$, $\nu$, and self-adjoint $A,B_{\xi}\in\cL(\cH)$ given as above, the following statements are equivalent:   
\begin{enumerate}[(a)] 
\item 
$\vec{\lambda}(A)\prec_w\int_\Xi\vec{\lambda}(B_\xi)\di\nu(\xi)$;\label{it:weakconvmaj} 


\item for every non-decreasing convex function $f:\mathbb{R}\to[0,\infty)$ and for every
unitarily invariant norm $\unorm{\cdot}$, 
\begin{align}
\unorm{f(A)} \le \int_\Xi\, \unorm{f(B_\xi)}\,\di\nu(\xi) \,.\label{eq:funitaryinvmaj}
\end{align}\label{it:funitaryinvmaj}
\end{enumerate}
\end{theorem}

\begin{proof}  
Assume~\eqref{it:weakconvmaj} and let $f$ be as in~\eqref{it:funitaryinvmaj}. We have 
\begin{align}
\vec{\lambda}(f(A))=f(\vec{\lambda}(A))\prec_w f\left(\int_\Xi\vec{\lambda}(B_\xi)\di\nu(\xi)\right)
\end{align}
thanks to~\cite[Prop.~4.1.4(2)]{hiai10b}. 
Since
\begin{align}
f\biggl(\int_\Xi \vec{\lambda}(B_\xi)\,\di\nu(\xi)\biggr)\le\int_\Xi f( \vec{\lambda}(B_\xi))\,\di\nu(\xi)
=\int_\Xi \vec{\lambda}(f(B_\xi))\,\di\nu(\xi),
\end{align}
we have $\vec{\lambda}(f(A))\prec_w \int_\Xi \vec{\lambda}(f(B_\xi))\,\di\nu(\xi)$. Since both sides of this relation are non-negative vectors, 
 applying the gauge function~$\Phi$ to them yields (see~\cite[Lemma 4.4.2]{hiai10b}) 
\begin{align}
{\unorm{f(A)}_\Phi}&\le\Phi\biggl(\int_\Xi \vec{\lambda}(f(B_\xi))\,\di\nu(\xi)\biggr)\\
&\leq \int_\Xi \Phi(\vec{\lambda}(f(B_\xi)))\, d\nu(\xi)=\int_\Xi{\unorm{f(B_\xi)}_\Phi}\,\di\nu(\xi)\ .
\end{align}
Hence~\eqref{it:funitaryinvmaj} holds.

To prove the converse, assume~\eqref{it:funitaryinvmaj}. Since $B_\xi$ is uniformly bounded, there is an $\alpha>0$ such that $A+\alpha I\geq 0$ and $B_\xi+\alpha I\geq 0$ for all $\xi\in\Xi$. We evaluate inequality~\eqref{eq:funitaryinvmaj} for $f(x):=\max\{x+\alpha,0\}$ and the Ky Fan norm $\unorm{\cdot}=\|\cdot\|_{(k)}$ to get 
\begin{align}
\sum_{i=1}^k (\lambda_i(A)+\alpha)\leq \sum_{i=1}^k \int_\Xi (\lambda_i(B_\xi)+\alpha)\, d\nu(\xi)\ .
\end{align}
Therefore, $\sum_{i=1}^k\lambda_i(A)\leq \sum_{i=1}^k \int_\Xi \lambda_i(B_\xi)\,d\nu(\xi)$, which implies~\eqref{it:weakconvmaj}.
\end{proof} 

\begin{remark}\label{rem:punar}
In the case where $A,B_\xi\in\cP$ for all $\xi\in\Xi$, conditions~\eqref{it:weakconvmaj} and~\eqref{it:funitaryinvmaj} are also equivalent to the statement
\begin{enumerate}[(a)]\setcounter{enumi}{2}
\item for every unitarily invariant norm $\unorm{\cdot}$,
\begin{align}
\unorm{A} \le \int_\Xi \unorm{B_\xi}\,\di\nu(\xi) \,.
\end{align}\label{it:unitaryinvnormmaj}
\end{enumerate}
\end{remark}
Indeed,~\eqref{it:funitaryinvmaj}~$\implies$~\eqref{it:unitaryinvnormmaj} is obvious by letting $f(x):=\max\{x,0\}$ in~\eqref{it:funitaryinvmaj}, and~\eqref{it:unitaryinvnormmaj}~$\implies$~\eqref{it:weakconvmaj} is seen by evaluating~\eqref{it:unitaryinvnormmaj} for the Ky Fan norm $\unorm{\cdot}=\|\cdot\|_{(k)}$. The assumption~$B_\xi\in\cP$  is essential for the latter implication. Statements~\eqref{it:weakconvmaj} and~\eqref{it:unitaryinvnormmaj} constitute a generalization of~($\star$) when applied to $|A|$ and $|B_{\xi}|$ for $A, B_{\xi}\in\cL(\cH)$.

\begin{theorem}\label{th:mainmajsimple}
With $\Xi$, $\nu$, and self-adjoint $A,B_{\xi}\in\cL(\cH)$ given as above, the following statements are equivalent:   
\begin{enumerate}[(a)] \setcounter{enumi}{3}
\item 
$\vec{\lambda}(A)\prec\int_\Xi\vec{\lambda}(B_\xi)\di\nu(\xi)$;\label{it:strconvmaj} 

\item inequality~\eqref{eq:funitaryinvmaj} holds for every convex function $f:\mathbb{R}\to[0,\infty)$ and for every
unitarily invariant norm $\unorm{\cdot}$. \label{it:strfunitaryinvmaj}
\end{enumerate}
\end{theorem}
\begin{proof}
Assume~\eqref{it:strconvmaj} and let $f$ be as in~\eqref{it:strfunitaryinvmaj}. 
It is obvious that $\vec{\lambda}(f(A))\approx f(\vec{\lambda}(A))$, where for 
$\a, \b\in\mathbb{R}^d$, $\a\approx\b$ means that the entries of~$\a$ coincide with those of~$\b$ up to a permutation.
Since Lemma~\ref{lm:convex} gives
\begin{align}
f(\vec{\lambda}(A))\prec_w f\biggl(\int_\Xi \vec{\lambda}(B_\xi)\,\di\nu(\xi)\biggr)\le\int_\Xi f( \vec{\lambda}(B_\xi))\,\di\nu(\xi)\, ,
\end{align}
we have 
\begin{align}
\unorm{f(A)}_\Phi &\le \Phi\biggl(\int_\Xi\, f(\vec{\lambda}(B_\xi))\,\di\nu(\xi)\biggr)\\
&\le \int_\Xi \Phi(f(\vec{\lambda}(B_\xi)))\,\di\nu(\xi)=\int_\Xi \unorm{f(B_\xi)}_\Phi\,\di\nu(\xi) \,.
\end{align}
Hence~\eqref{it:strfunitaryinvmaj} holds. 

Conversely, if~\eqref{it:strfunitaryinvmaj} is satisfied, then by Theorem~\ref{th:mainweaksimple} we have $\vec{\lambda}(A)\prec_w\int_\Xi \vec{\lambda}(B_\xi)\,\di\nu(\xi)$. Hence, to prove~\eqref{it:strconvmaj}, it suffices to show that $\tr A\geq \int_\Xi \tr B_\xi\,\di\nu(\xi)$. Choose an $\alpha>0$ such that $A\leq \alpha I$ and $B_\xi\leq \alpha I$ for all $\xi\in\Xi$. For $f(x):=\max\{\alpha-x,0\}$ and the trace norm $\unorm{\cdot}=\|\cdot\|_1$, condition~\eqref{it:strfunitaryinvmaj} implies that 
\begin{align}
\tr(\alpha I-A) \leq \int_\Xi \tr(\alpha I-B_\xi)\,\di\nu(\xi)\ ,
\end{align}
giving the desired inequality. 
\end{proof}

\section{Weak log-majorization with integral average}
\label{sec:main}
In the following, we assume that $A, B_\xi\in\cP$ for all $\xi\in\Xi$. The above equivalent conditions~\eqref{it:weakconvmaj},~\eqref{it:funitaryinvmaj} of Theorem~\ref{th:mainweaksimple} and~\eqref{it:unitaryinvnormmaj} of Remark~\ref{rem:punar} correspond to weak majorization, and~\eqref{it:strconvmaj},~\eqref{it:strfunitaryinvmaj} of Theorem~\ref{th:mainmajsimple} correspond to majorization. We now consider stronger conditions than those, corresponding to (weak) log-majorization. We have the following chain of implications, where the last condition~\eqref{it:fourthconditionchain} is~\eqref{it:funitaryinvmaj} of Theorem~\ref{th:mainweaksimple} and we use as in~\eqref{eq:conventionintegralvec} the convention 
\begin{align}
\int_\Xi \log \vec{\lambda}(B_\xi)\,\di\nu(\xi)
&:=\biggl(\int_\Xi \log \lambda_1(B_\xi)\,\di\nu(\xi),\dots,\int_\Xi \log\lambda_d(B_\xi)\,\di\nu(\xi)\biggr) . 
\end{align}

\begin{proposition}\label{lm:chainofimplications}
With $\Xi$, $\nu$, $A$ and $B_{\xi}$ {given} as above, consider the following statements:
\begin{enumerate}[(1)]
\item
$\vec{\lambda}(A) \prec_{\log}\exp \int_\Xi \log \vec{\lambda}(B_\xi)\,\di\nu(\xi)$\label{it:firstconditionchain}
\item
$\vec{\lambda}(A) \prec_{w\log} \exp \int_\Xi \log \vec{\lambda}(B_\xi)\,\di\nu(\xi)$\label{it:secondconditionchain}
\item
$\vec{\lambda}(A)\prec_{w\log}\int_\Xi \,\vec{\lambda}(B_\xi)\,\di\nu(\xi)$
\label{it:thirdconditionchain}
\item
$\vec{\lambda}(A)\prec_w\int_\Xi\vec{\lambda}(B_\xi)\,\di\nu(\xi)$.
\label{it:fourthconditionchain}
\end{enumerate}
Then~\eqref{it:firstconditionchain}~$\implies$~\eqref{it:secondconditionchain}~$\implies$~\eqref{it:thirdconditionchain}~$\implies$~\eqref{it:fourthconditionchain}. 
\end{proposition}
\begin{proof}
The implication~\eqref{it:firstconditionchain}~$\implies$~\eqref{it:secondconditionchain} is trivial. Implication \eqref{it:secondconditionchain}~$\implies$~\eqref{it:thirdconditionchain} follows by Jensen's inequality,
and~\eqref{it:thirdconditionchain}~$\implies$~\eqref{it:fourthconditionchain} 
follows by Lemma~\ref{lm:convex}. 
\end{proof}

The following theorem  constitutes part of our main results and characterizes the second condition in this chain. We will give a similar characterization of the first condition  in Theorems~\ref{th:main} and~\ref{th:mainadditional} below. 

\begin{theorem}\label{th:mainweak}
With $\Xi$, $\nu$, $A$ and $B_{\xi}$ {given} as above, the following statements are equivalent:
\begin{enumerate}[(i)]
\item 
$\vec{\lambda}(A) \prec_{w\log} \exp \int_\Xi \log \vec{\lambda}(B_\xi)\,\di\nu(\xi)$, i.e.,\\
 $\log \vec{\lambda}(A) \prec_w \int_\Xi \log \vec{\lambda}(B_\xi)\,\di\nu(\xi)$;\label{it:loglambdaweak} 
 
\item for every continuous non-decreasing function $f:[0,\infty)\to[0,\infty)$ such that
$x \mapsto \log f(e^x)$ is convex on~$\bR$, and for every unitarily invariant norm $\unorm{\cdot}$,\label{it:nondecr}
\begin{align}
\unorm{f(A)} \le \exp\int_\Xi \log \unorm{ f(B_\xi) }\, \di\nu(\xi) \,;\label{F-II}
\end{align}  
\item for every continuous non-decreasing function $g:[0,\infty)\to[0,\infty)$ such that
$x \mapsto g(e^x)$ is convex on $\bR$, and for every unitarily invariant norm $\unorm{\cdot}$,
\begin{align}
\unorm{g(A)} \le\int_\Xi \unorm{g(B_\xi)}\, \di\nu(\xi)\,  \,.\label{eq:gIII}
\end{align}\label{it:gnondecrcond}
\end{enumerate}
\end{theorem}  

When $\Xi$ is a one-point set, Theorem~\ref{th:mainweak}
 reduces to \cite[Prop. 4.4.13]{hiai10b}, except condition~\eqref{it:nondecr}. The proof of \eqref{it:nondecr}~$\implies$ \eqref{it:loglambdaweak} given below
implies the next proposition, which appears to be a new characterization of weak log-majorization.
\begin{proposition}
\label{pr:charweaklog}
For $A,B\in\cP$, 
$\vec{\lambda}(A) \prec_{w\log} \vec{\lambda}(B)$ if and only if
\begin{align}
\|A^p\|_{(k)}\le\|B^p\|_{(k)},\qquad p>0,\ \ k \in [d] .
\end{align}
\end{proposition}


Before presenting the proof of Theorem~\ref{th:mainweak}, we discuss the convexity conditions appearing in~\eqref{it:nondecr} and~\eqref{it:gnondecrcond}. We will use the following properties 
 in the proof of Theorem~\ref{th:mainweak} and again in Section~\ref{sec:logmajorizationintegral}. 
\begin{lemma}\label{lm:main1}
Let $f:(0,\infty)\to[0,\infty)$ be a continuous function such that $x\mapsto \log f(e^x)$ is convex on $\bR$. Then:
\begin{enumerate}[(1)]
\item $f(x)>0$ for any $x>0$ unless $f\equiv0$. \label{lm:main1(a)}
\item $f(0^+):=\lim_{x\searrow0}f(x)$ exists in $[0,\infty]$, and if $f(0^+)<\infty$ then $f$ is non-decreasing on
$(0,\infty)$. \label{lm:main1(b)}
\end{enumerate}
Similarly, if $g:(0,\infty)\to[0,\infty)$ is a  function such that $x\mapsto g(e^x)$ is convex on $\bR$, then $g$ is automatically continuous on $(0,\infty)$,   $g(0^+):=\lim_{x\searrow0}g(x)$ exists in $[0,\infty]$, and if $g(0^+)<\infty$ then $g$ is non-decreasing on $(0,\infty)$.
\end{lemma}

In particular,~\eqref{lm:main1(b)} implies that
a non-decreasing continuous function $f:(0,\infty)\rightarrow[0,\infty)$ with the 
property that $x\mapsto \log f(e^x)$ is convex on~$\mathbb{R}$ extends to a continuous function $f:[0,\infty)\rightarrow [0,\infty)$. This corresponds to~\eqref{it:nondecr} of Theorem~\ref{th:mainweak}. Analogously, if $g:(0,\infty)\rightarrow [0,\infty)$ is continuous and non-decreasing, and $x\mapsto g(e^x)$ is convex on~$\mathbb{R}$, then~$g$ extends to a continuous function $g:[0,\infty)\rightarrow [0,\infty)$. Such functions appear in~\eqref{it:gnondecrcond} of Theorem~\ref{th:mainweak}. In Section~\ref{sec:logmajorizationintegral}, we will drop the assumption that $f,g$ are non-decreasing and instead consider majorization instead of weak majorization. 

For instance, for any $\alpha\ge0$ and any $p>0$, $f(x):=(\alpha+x)^p$ is an increasing function on $[0,\infty)$
such that $\log f(e^x)$ is convex on $\bR$. For any $p>0$, $g(x):=\log(1+x^p)$ is an increasing function on
$[0,\infty)$ such that $g(e^x)$ is convex on $\bR$ but $\log g(e^x)$ is concave on $\bR$.

\begin{proof}
\eqref{lm:main1(a)}\enspace
Let $h(x):=\log f(e^x)$, $x\in\bR$. If the conclusion is not true, then there is an $\alpha>0$ such that
$f(\alpha)=0$ and $f(x)>0$ for $x\in(\alpha-\delta,\alpha)$ or $x\in(\alpha,\alpha+\delta)$ for some $\delta>0$.
Since $\lim_{x\to\log \alpha}h(x)=\log f(\alpha)=-\infty$, $h$ cannot be convex around~$\log \alpha$.

\hspace{-0.08cm}\eqref{lm:main1(b)}\enspace\hspace{-0.09cm}
From the convexity of $h$ on $\bR$ it follows that $h(-\infty):=\lim_{x\to-\infty}h(x)$ exists in $[-\infty,\infty]$. This
implies that $f(0^+)$ exists in $[0,\infty]$. Unless $h$ is non-decreasing on $\bR$, the convexity of $h$ implies
that $h(-\infty)=\infty$ and so $f(0^+)=\infty$. 

The proof of the statements for $g$ is similar and omitted here. 
\end{proof}

\subsection*{Proof of Theorem~\ref{th:mainweak}}

In comparing conditions~\eqref{it:nondecr} and~\eqref{it:gnondecrcond}, the convexity of $\log f(e^x)$ in~\eqref{it:nondecr} is stronger than
the convexity of $f(e^x)$ in~\eqref{it:gnondecrcond}. Correspondingly, the conclusion of~\eqref{it:nondecr} is stronger than
that of~\eqref{it:gnondecrcond}. Hence it is not clear how to pass directly between~\eqref{it:nondecr} and~\eqref{it:gnondecrcond}.
The proof is thus split into four parts, corresponding to the implications~\eqref{it:loglambdaweak}~$\implies$~\eqref{it:nondecr}, \eqref{it:nondecr}~$\implies$~\eqref{it:loglambdaweak}, \eqref{it:loglambdaweak}~$\implies$~\eqref{it:gnondecrcond}, and \eqref{it:gnondecrcond}~$\implies$~\eqref{it:loglambdaweak}.

\begin{proof}[Proof of \eqref{it:loglambdaweak}~$\implies$~\eqref{it:nondecr}]
First, assume that $A,B_\xi\in\cPp$ and $B_\xi\ge\eps I$ for all $\xi\in\Xi$ with some $\eps>0$. Because
$f$ is non-increasing on $[0,\infty)$, we have 
\begin{align}
\log \vec{\lambda}(f(A))&=\log f(e^{\log \vec{\lambda}(A)})\ .\label{eq:equalityordering}
\end{align}
Since $\log f(e^x)$ is convex on $\bR$, Lemma~\ref{lm:convex} yields
\begin{align}
\log f(e^{\log \vec{\lambda}(A)})
&\prec_w\log f\biggl(\exp\int_\Xi\log \vec{\lambda}(B_\xi)\,\di\nu(\xi)\biggr) \\
&\le\int_\Xi\log f( \vec{\lambda}(B_\xi))\,\di\nu(\xi)\label{eq:initialstepweakproof}
\end{align}
 from condition~\eqref{it:loglambdaweak}.  Therefore, we have with~\eqref{eq:equalityordering}
\begin{align}
\vec{\lambda}(f(A))\prec_w\exp\int_\Xi\log f(\vec{\lambda}(B_\xi))\,\di\nu(\xi)
\end{align}
so that
\begin{equation}\label{F-1}
\unorm{f(A)}_{\Phi} = \Phi(\vec{\lambda}(f(A)))
\le\Phi\biggl(\exp\int_\Xi\log f(\vec{\lambda}(B_\xi))\,\di\nu(\xi)\biggr) \,.
\end{equation}

By Lemma \ref{lm:main1}\,\eqref{lm:main1(a)} we may assume that $f(x)>0$ for any $x>0$,
so the function $\xi\mapsto\log f(\lambda_i(B_\xi))$, $i \in [d]$, as well as
$\xi\mapsto \log \unorm{f(B_\xi)}_{\Phi}$ are bounded and continuous on $\Xi$. Since, moreover,
$\nu(\Xi)=1=\sup\{\nu(K):K\subset\Xi\ \mbox{compact}\}$ due to the $\sigma$-compactness of $\Xi$,
a standard compactness argument shows that there are ${\xi_l}^{(m)}\in\Xi$ and ${\beta_l}^{(m)}>0$
for $l \in [m]$ and $m\in\bN$ with $\sum_{l=1}^m {\beta_l}^{(m)}=1$ such that
\begin{align}
\int_\Xi\log f(\lambda_i(B_\xi))\,\di\nu(\xi)
&=\lim_{m\to\infty}\sum_{l=1}^m\beta_l^{(m)}\log f\Bigl(
\lambda_i\Bigl(B_{\xi_l^{(m)}}\Bigr)\Bigr) \,,\qquad i \in [d], \\
\int_\Xi\log \unorm{f(B_\xi)}_{\Phi} \,\di\nu(\xi)
&=\lim_{m\to\infty} \sum_{l=1}^m \beta_l^{(m)} \log
 \unormb{f\Bigl(B_{\xi_l^{(m)}}\Bigr)}_{\Phi} \,.
\end{align}
Therefore,
\begin{align}
\Phi\biggl(\exp\int_\Xi\log f(\vec{\lambda}(B_\xi))\,\di\nu(\xi)\biggr)
&=\lim_{m\to\infty}\Phi\Biggl(\prod_{l=1}^m f\Bigl(
\vec{\lambda}\Bigl(B_{\xi_l^{(m)}}\Bigr)\Bigr)^{\beta_l^{(m)}}\Biggr), \\
\exp\int_\Xi\log \unorm{f(B_\xi)}_{\Phi} \,\di\nu(\xi)
&=\lim_{m\to\infty}\prod_{l=1}^m \unormb{f\Bigl(B_{\xi_l^{(m)}}\Bigr)}_{\Phi}^{\beta_l^{(m)}} \,.
\end{align}
Since Lemma~\ref{lm:hoelder} implies that
\begin{align}
\Phi\Biggl(\prod_{l=1}^m f\Bigl(
\vec{\lambda}\Bigl(B_{\xi_l^{(m)}}\Bigr)\Bigr)^{\beta_l^{(m)}}\Biggr)
&\le\prod_{l=1}^m\Phi^{\beta_l^{(m)}}
\Bigl(f\Bigl(\vec{\lambda}\Bigl(B_{\xi_l^{(m)}}\Bigr)\Bigr)\Bigr) \\
&=\prod_{l=1}^m\Phi^{\beta_l^{(m)}}
\Bigl(\vec{\lambda}\Bigl(f\Bigl(B_{\xi_l^{(m)}}\Bigr)\Bigr)\Bigr)\\
&=\prod_{l=1}^m \unormb{f\Bigl(B_{\xi_l^{(m)}}\Bigr)}_{\Phi}^{\beta_l^{(m)}} \,,
\end{align}
we have
\begin{equation}\label{F-2}
\Phi\biggl(\exp\int_\Xi\log f(\vec{\lambda}(B_\xi))\,\di\nu(\xi)\biggr)
\le\exp\int_\Xi\log \unorm{f(B_\xi)}_{\Phi} \,\di\nu(\xi) \,.
\end{equation}
Combining \eqref{F-1} and \eqref{F-2} gives inequality \eqref{F-II}.

Next, consider the general case where $A,B_\xi\in\cP$. For any $\eps>0$, since
\begin{align}
\prod_{i=1}^k\lambda_i(A)&\le\prod_{i=1}^k\exp\int_\Xi\log\lambda_i(B_\xi)\,\di\nu(\xi) \\
&<\prod_{i=1}^k\exp\int_\Xi\log(\lambda_i(B_\xi)+\eps)\,\di\nu(\xi),\qquad k \in [d],  
\end{align}
one can choose a $\delta_\eps\in(0,\eps)$ such that
\begin{align}
\prod_{i=1}^k(\lambda_i(A)+\delta_\eps)
\le\prod_{i=1}^k\exp\int_\Xi\log(\lambda_i(B_\xi)+\eps)\,\di\nu(\xi), {\qquad k \in [d],}\label{eq:epsilondeltaapprox}
\end{align}
i.e., we have the weak log-majorization
$\vec{\lambda}(A+\delta_\eps I)\prec_{w\log}\exp\int_\Xi\log \vec{\lambda}(B_\xi+\eps I)\,\di\nu(\xi)$.
 By applying the first case to $A+\delta_\eps I$ and $B_\xi+\eps I$, we have
\begin{equation}\label{F-3}
\unorm{f(A+\delta_\eps I)}_{\Phi} \le 
\exp\int_\Xi\log \unorm{f(B_\xi+\eps I)}_{\Phi} \ \di\nu(\xi).
\end{equation}
Since $\log \unorm{f(B_\xi+\eps I)}_{\Phi} \searrow\log \unorm{f(B_\xi)}_{\Phi}$ for every $\xi\in\Xi$ as $\eps\searrow0$,
the monotone convergence theorem gives
\begin{align}
\int_\Xi\log \unorm{f(B_\xi+\eps I)}_{\Phi} \,\di\nu(\xi)\searrow\int_\Xi\log \unorm{f(B_\xi)}_{\Phi} \ \di\nu(\xi).
\end{align}
Hence letting $\eps\searrow0$ in \eqref{F-3} gives the desired inequality.
\end{proof}


\begin{proof}[Proof of \eqref{it:nondecr}~$\implies$~\eqref{it:loglambdaweak}]
First, assume that $B_\xi\in\cPp$ for all $\xi\in\Xi$. Assume~\eqref{it:nondecr}, and for every
{$k \in [d]$} we prove that
\begin{equation}\label{F-4}
\prod_{i=1}^k\lambda_i(A)\le\prod_{i=1}^k\exp\int_\Xi\log\lambda_i(B_\xi)\,\di\nu(\xi).
\end{equation}
Since \eqref{F-4} is obvious if $\lambda_k(A)=0$, we may assume that $\lambda_k(A)>0$.
Applying inequality~\eqref{F-II} in~\eqref{it:nondecr} to $\unorm{\cdot}=\|\cdot\|_{(k)}$, and
$f(x)=x^p$ for each $p>0$ (which obviously satisfies the condition in~\eqref{it:nondecr}), we have
\begin{align}
\|A^p\|_{(k)}\le\exp\int_\Xi\log\|B_\xi^p\|_{(k)}\,\di\nu(\xi),
\end{align}
{i.e.,}
\begin{align}
\log\sum_{i=1}^k\lambda_i^p(A)\le\int_\Xi\log\sum_{i=1}^k\lambda_i^p(B_\xi)\,\di\nu(\xi).
\end{align}
Therefore,
\begin{equation}\label{F-5}
{\frac{1}{p}}\log\Biggl({\frac{1}{k}}\sum_{i=1}^k\lambda_i^p(A)\Biggr)
\le\int_\Xi{\frac{1}{p}}\log\Biggl({\frac{1}{k}}\sum_{i=1}^k\lambda_i^p(B_\xi)\Biggr)\,\di\nu(\xi).
\end{equation}
Since, for $a_i>0$, the function $p>0\mapsto\log\bigl({\frac{1}{k}}\sum_{i=1}^ka_i^p\bigr)$ is
convex, we find that, as $p\searrow0$, 
\begin{align}
{\frac{1}{p}}\log\Biggl({\frac{1}{k}} \sum_{i=1}^k\lambda_i^p(A)\Biggr)
&\searrow {\frac{1}{k}}\sum_{i=1}^k\log\lambda_i(A), \label{F-7}\\
{\frac{1}{p}}\log\Biggl({\frac{1}{k}}\sum_{i=1}^k\lambda_i^p(B_\xi)\Biggr)
&\searrow {\frac{1}{k}}\sum_{i=1}^k\log\lambda_i(B_\xi)\label{F-8}
\end{align}
for every $\xi\in\Xi$. 
This relies on the fact that $x\mapsto f(x)/x$ is non-decreasing if $f$ is convex and $f(0)=0$. 
Hence the monotone convergence theorem yields
\begin{align}
\int_\Xi{\frac{1}{p}}\log\Biggl({\frac{1}{k}}\sum_{i=1}^k\lambda_i^p(B_\xi)\Biggr)\,\di\nu(\xi)
\searrow{\frac{1}{k}}\int_\Xi\sum_{i=1}^k\log\lambda_i(B_\xi)\,\di\nu(\xi)\qquad
\mbox{as $p\searrow0$}.
\end{align}
Therefore, by letting $p\searrow0$ in \eqref{F-5} we have
\begin{align}
\sum_{i=1}^k\log\lambda_i(A)\le\int_\Xi\sum_{i=1}^k\log\lambda_i(B_\xi)\,\di\nu(\xi),
\end{align}
implying \eqref{F-4}.

Next, consider the general case where $B_\xi\in\cP$ for all $\xi\in\Xi$. Since the inequality~\eqref{F-II} in~\eqref{it:nondecr} holds with $B_\xi+\eps I$ instead of $B_\xi$ for any $\eps>0$, the above case implies
that
\begin{align}
\prod_{i=1}^k\lambda_i(A)\le\prod_{i=1}^k\exp\int_\Xi\log(\lambda_i(B_\xi)+\eps)\,\di\nu(\xi),
\qquad k \in [d].
\end{align}
Letting $\eps\searrow0$ gives \eqref{F-4} so that~\eqref{it:loglambdaweak} follows.
\end{proof}


\begin{proof}[Proof of \eqref{it:loglambdaweak} $\implies$~\eqref{it:gnondecrcond}]
Assume first that $A, B_\xi\in\cPp$
and $B_\xi\geq \varepsilon I$ for all $\xi\in\Xi$ with some $\varepsilon>0$. 
Since~\eqref{it:loglambdaweak} means that 
\begin{align}
\vec{\lambda}(\log A)=\log \vec{\lambda}(A) \prec_w\int_\Xi \log \vec{\lambda}(B_\xi)\,\di\nu(\xi)=\int_\Xi \vec{\lambda}(\log B_\xi)\,\di\nu(\xi)\, ,
\end{align}
one can apply~\eqref{it:weakconvmaj}~$\implies$~\eqref{it:funitaryinvmaj} of Theorem~\ref{th:mainweaksimple} to $\log A$, $\log B_\xi$ and $f(x):=g(e^x)$, where $g$ is as in~\eqref{it:gnondecrcond}. Inequality~\eqref{eq:gIII} then immediately follows. For the general case where $A, B_\xi\in\cP$, for any $\varepsilon>0$ choose a $\delta_\varepsilon\in (0,\varepsilon)$ satisfying~\eqref{eq:epsilondeltaapprox}. Since the above case gives $\unorm{g(A+\delta_\varepsilon I)}\leq \int_\Xi \unorm{g(B_\xi+\varepsilon I)}\,\di\nu(\xi)$, we have~\eqref{eq:gIII}
by letting~$\varepsilon\searrow0$.

\end{proof}


\begin{proof}[Proof of \eqref{it:gnondecrcond}~$\implies$~\eqref{it:loglambdaweak}]
For $k\in [d]$ let $\unorm{\cdot}=\|\cdot\|_{(k)}$ and $g(x):=\log(1+\eps^{-1}x)$ where $\eps>0$; then $g$ satisfies
the condition in~\eqref{it:gnondecrcond}. Since
\begin{align}
\|g(A)\|_{(k)} &=\sum_{i=1}^k\log(\eps+\lambda_i(A))-k\log\eps, \\
\int_\Xi\|g(B_\xi)\|_{(k)}\,\di\nu(\xi)
&= \sum_{i=1}^k\int_\Xi\log(\eps+\lambda_i(B_\xi))\,\di\nu(\xi)-k\log\eps,
\end{align} 
inequality~\eqref{eq:gIII} implies that
\begin{align}
\sum_{i=1}^k\log(\eps+\lambda_i(A))
\le\sum_{i=1}^k\int_\Xi\log(\eps+\lambda_i(B_\xi))\,\di\nu(\xi).
\end{align}
Letting $\eps\searrow0$ gives
\begin{align}
\sum_{i=1}^k\log\lambda_i(A)\le\sum_{i=1}^k\int_\Xi\log\lambda_i(B_\xi)\,\di\nu(\xi),
\end{align}
and hence~\eqref{it:loglambdaweak} follows.
\end{proof}

\section{Log-majorization with integral average\label{sec:logmajorizationintegral}}

Consider the strongest condition~\eqref{it:firstconditionchain} in the chain of implications 
in Proposition~\ref{lm:chainofimplications}. 
Our first main result concerning this condition is the following. 
\begin{theorem} \label{th:main}
With $\Xi$, $\nu$, $A$ and $B_{\xi}$ {given} as above, the following statements are equivalent:
\begin{enumerate}[(I)]
\item
{$\vec{\lambda}(A) \prec_{\log} \exp \int_\Xi \log \vec{\lambda}(B_\xi)\,\di\nu(\xi)$, i.e.,
$\log \vec{\lambda}(A) \prec \int_\Xi \log \vec{\lambda}(B_\xi)\,\di\nu(\xi)$;}\label{it:loglambda}
\item for every continuous  function $f:(0,\infty)\to[0,\infty)$ such that
$x \mapsto \log f(e^x)$ is convex on~$\bR$, and for every unitarily invariant norm $\unorm{\cdot}$,\label{it:fcond}
\begin{align}
\unorm{f(A)} \le \exp\int_\Xi \log \unorm{ f(B_\xi) }\, \di\nu(\xi) \, .\label{F-IInonweak} 
\end{align}
In this statement,  we extend $f$ to $[0,\infty)$ by continuity
and for any unitarily invariant norm $\unorm{\cdot}$ use the
convention  $\unorm{f(A)}=\infty$ when $f(0^+)=\infty$ and $A\in\cP$ is not invertible. 
\end{enumerate}
\end{theorem} 
The proof requires a few auxiliary results. We first show that the right-hand side of~\eqref{F-IInonweak} is well-defined.

\begin{lemma}\label{lem:mainlemmamod}
Let $f:(0,\infty)\to[0,\infty)$ be a continuous function such that $x\mapsto \log f(e^x)$ is convex on $\bR$.
Then (with the extension and convention as stated in Theorem~\ref{th:main})
 $\int_\Xi\log\unorm{f(B_\xi)}\,\di\nu(\xi)$ exists in $[-\infty,\infty]$. \label{lm:main1(c)}
\end{lemma}
\begin{proof}
If $f(0^+)<\infty$, then by Lemma~\ref{lm:main1}~\eqref{lm:main1(b)} (and the uniform boundedness of $B_\xi$) we have
$\sup_\xi\unorm{f(B_\xi)}<\infty$, so the integral $\int_\Xi\log\unorm{f(B_\xi)}\,\di\nu(\xi)$ exists in $[-\infty,\infty)$.
If $f(0^+)=\infty$, then one can choose $a>0$ and $b\in\bR$ such that $\log f(e^x)\ge-ax+b$ on~$\bR$ by the convexity assumption.  This in turn
implies that $f(x)\ge e^bx^{-a}$ on $(0,\infty)$. Hence we have $\inf_\xi\unorm{f(B_\xi)}>0$, so the integral exists
in $(-\infty,\infty]$.
\end{proof}

Before proving the theorem we give another lemma.
\begin{lemma}\label{lm:main2}
Let $\a,\b\in\bR_+^d$ be such that $a_1\ge\dots\ge a_d$ and $b_1\ge\dots\ge b_d$, and assume that
$\a\prec_{\log}\b$. Furthermore, let $\b^{(m)}\in\bR_+^d$, $m\in\bN$, be such that
$b^{(m)}_1\ge\dots\ge b^{(m)}_d>0$ and $\b^{(m)}\searrow\b$ as $m\to\infty$. Then there exist $m_0\in\bN$
and $\a^{(m)}\in\bR_+^d$ for $m\ge m_0$ such that $a^{(m)}_1\ge\dots\ge a^{(m)}_d>0$, $\a^{(m)}\to\a$ as
$m\to\infty$, and
$$
\a\le\a^{(m)}\prec_{\log}\b^{(m)},\qquad m\ge m_0 \,.
$$

\end{lemma}

\begin{proof}
Assume that $\a\prec_{\log}\b$. The proof is divided into two cases. First, assume that $a_d>0$ (hence $b_d>0$
as well). For each $m\in\bN$, since $\log\a,\log\b\in\bR^d$ and $\log\a\prec\log\b\le\log\b^{(m)}$ so that
$\log\a\prec_w\log\b^{(m)}$, it follows from \cite[Proposition 4.1.3]{hiai10b} that there exists a $\c^{(m)}\in\bR^d$
such that
\begin{align}
\log\a\le\c^{(m)}\prec\log\b^{(m)} \,.
\end{align}
Now define $\a^{(m)}:=\exp\c^{(m)}$; then $\a\le\a^{(m)}\prec_{\log}\b^{(m)}$.
It remains to prove that $\a^{(m)}\to\a$. For this, note that
\begin{align}
\sum_{i=1}^d\bigl(c_i^{(m)}-\log a_i\bigr)=\sum_{i=1}^d\bigl(\log b_i^{(m)}-\log a_i\bigr)
\longrightarrow\sum_{i=1}^d(\log b_i-\log a_i)=0
\end{align}
as $m\to\infty$. Therefore, we have $c_i^{(m)}\to\log a_i$ so that $a_i^{(m)}\to a_i$ for all $i\in[d]$.

Secondly, assume that $a_d=0$ (hence $b_d=0$ as well). Assume that
\begin{align*}
a_1\ge\dots\ge a_r&>0=a_{r+1}=\dots=a_d \,, \\
b_1\ge\dots\ge b_s&>0=b_{s+1}=\dots=b_d \,.
\end{align*}
Since $0<\prod_{i=1}^r a_i\le\prod_{i=1}^r b_i$, we have $r\le s$. For each $m\in\bN$ define
\begin{align}
\alpha^{(m)}:=\Biggl({\frac{\prod_{i=1}^{s+1}b_i^{(m)}}{\prod_{i=1}^ra_i}}\Biggr)^{\frac{1}{s-r+1}}>0 \,.
\end{align}
Since $b_{s+1}^{(m)}\to b_{s+1}=0$ and so $\alpha^{(m)}\to0$ as $m\to\infty$, choose an $m_0$ such that
$\alpha^{(m)}\le\min\{a_r,b_s\}$ for all $m\ge m_0$. Define for $m\ge m_0$,
\begin{align}
\a^{(m)}:=\Bigl(a_1,\dots,a_r,\underbrace{\alpha^{(m)},\dots,\alpha^{(m)}}_{s-r+1},
b_{s+2}^{(m)},\dots,b_d^{(m)}\Bigr) \,.
\end{align}
Since
\begin{align}
\alpha^{(m)}={\frac{\prod_{i=1}^{s+1}b_i^{(m)}}{\prod_{i=1}^ra_i\cdot\bigl(\alpha^{(m)}\bigr)^{s-r}}}
\ge{\frac{\prod_{i=1}^rb_i}{\prod_{i=1}^ra_i}}\biggl({\frac{b_s}{\alpha^{(m)}}}\biggr)^{s-r}b_{s+1}^{(m)}
\ge b_{s+1}^{(m)} \,,
\end{align}
we find that $\a^{(m)}$ is in decreasing order. We furthermore have $\a^{(m)}\to\a$ and
\begin{align}
\prod_{i=1}^ra_i\cdot(\alpha^{(m)})^k\le\prod_{i=1}^rb_i\cdot b_s^k\le\prod_{i=1}^{r+k}b_i^{(m)} \,,
\qquad1\le k\le s-r,
\end{align}
\begin{align}
\prod_{i=1}^ra_i\cdot(\alpha^{(m)})^{s-r+1}=\prod_{i=1}^{s+1}b_i^{(m)} \,,
\end{align}
so that $\a\le\a^{(m)}\prec_{\log}\b^{(m)}$ follows.
\end{proof}


\begin{proof}[Proof of  \eqref{it:loglambda}~$\implies$~\eqref{it:fcond}]
First, assume that $A,B_\xi\in\cPp$ and $B_\xi\ge\eps I$ for all $\xi\in\Xi$ with some $\eps>0$. Since $\vec{\lambda}(f(A))\approx f(\vec{\lambda}(A))$, the corresponding part of the proof of Theorem~\ref{th:mainweak}
can be adopted with the slight modification that
\begin{align}
\log \vec{\lambda}(f(A))&\approx\log f(e^{\log \vec{\lambda}(A)})
\end{align}
instead of~\eqref{eq:equalityordering} because the assumption that $f$ is non-decreasing has been dropped.

Next, consider the general case where $A,B_\xi\in\cP$. With $0<\eps_m\searrow0$, we have 
\begin{align}
\prod_{i=1}^k\lambda_i(A)&\le\prod_{i=1}^k\exp\int_\Xi\log\lambda_i(B_\xi)\,\di\nu(\xi) \\
&<\prod_{i=1}^k\exp\int_\Xi\log(\lambda_i(B_\xi)+\eps_m)\,\di\nu(\xi) \,,\qquad k \in [d]\ .
\end{align}
Since
$\int_\Xi\log(\vec{\lambda}(B_\xi)+\eps_m)\,\di\nu(\xi)\searrow\int_\Xi\log\vec{\lambda}(B_\xi)\,\di\nu(\xi)$ as
$m\to\infty$ by the monotone convergence theorem, by Lemma \ref{lm:main2} one can choose $\a^{(m)}$,
$m\ge m_0$, such that $a_1^{(m)}\ge\dots\ge a_d^{(m)}>0$, $\a^{(m)}\to\vec{\lambda}(A)$ and
\begin{align}
\a^{(m)}\prec_{\log}\exp\int_\Xi\log(\vec{\lambda}(B_\xi+\eps_mI))\,\di\nu(\xi) \,.
\end{align}
Choosing $A^{(m)}\in\cPp$ with $\vec{\lambda}(A^{(m)})=\a^{(m)}$ and applying the first case to $A^{(m)}$ and $B_\xi+\eps_mI$, we have
\begin{equation}\label{F-3x}
\unorm{f(A^{(m)})}_{\Phi} \le \exp\int_\Xi\log \unorm{f(B_\xi+\eps_mI)}_{\Phi} \ \di\nu(\xi)\,,\qquad m\ge m_0. 
\end{equation}
When $f(0^+)<\infty$ and hence $f$ is non-decreasing on $(0,\infty)$ by Lemma
\ref{lm:main1}\,\eqref{lm:main1(b)}, note that 
\begin{align}
\unorm{f(A^{(m)})}_{\Phi}=\Phi(f(\a^{(m)}))\longrightarrow\Phi(f(\vec{\lambda}(A)))=\unorm{f(A)}_{\Phi}
\end{align}
and similarly $\unorm{f(B_\xi+\eps_mI)}_{\Phi} \to \unorm{f(B_\xi)}_{\Phi}$ for every $\xi\in\Xi$ as $m\to\infty$.
Since $\xi\mapsto\unorm{f(B_\xi+\eps_mI)}_{\Phi}$ is uniformly bounded above (so
$-\log\unorm{f(B_\xi+\eps_mI)}_{\Phi}$ is uniformly bounded below), Fatou's lemma yields
\begin{align}
\limsup_{m\to\infty}\int_\Xi\log \unorm{f(B_\xi+\eps_mI)}_{\Phi} \,\di\nu(\xi)\le
\int_\Xi\log \unorm{f(B_\xi)}_{\Phi} \ \di\nu(\xi) \,,
\end{align}
and therefore, letting $m\rightarrow\infty$ in \eqref{F-3x} gives inequality \eqref{F-IInonweak}. Finally, when $f(0^+)=\infty$,
we may assume that $\int_\Xi\log\unorm{f(B_\xi)}_\Phi\,\di\nu(\xi)<\infty$. In this case, $f$ is decreasing on
$(0,\delta)$ for some $\delta>0$. 
We will argue below that  there are constants $\alpha,\beta>0$ such that
\begin{align}
\alpha\le\unorm{f(B_\xi+\eps_mI)}_\Phi\le\unorm{f(B_\xi)}_{\Phi}+\beta\label{eq:upperlowerconstants}
\end{align}
for all $\xi\in\Xi$ and $m\ge m_0$. 
Since $\int_\Xi\log(\unorm{f(B_\xi)}_{\Phi}+\beta)\,\di\nu(\xi)<\infty$, the
Lebesgue convergence theorem can be used to get~\eqref{F-IInonweak} by taking the limit of~\eqref{F-3x}.

It remains to show~\eqref{eq:upperlowerconstants}. By the uniform boundedness of the operators~$B_\xi$, there is a constant~$\gamma>0$ such that 
\begin{align}
0<B_\xi+\eps_m I\leq \gamma I\, ,\qquad\xi\in\Xi\,, m\geq m_0\ .\label{eq:upperlowerboundbxi}
\end{align} Because $f$ is decreasing on $(0,\delta)$ and $f(x)>0$ for all $x>0$ (see Lemma~\ref{lm:main1}~\eqref{lm:main1(a)}), this implies that $\xi \mapsto \unorm{f(B_\xi+\eps_mI)}_\Phi$ is uniformly bounded from below, as claimed in~\eqref{eq:upperlowerconstants}.  Observe that the upper bound in~\eqref{eq:upperlowerconstants} is trivial for 
$B_\xi\in\cP\backslash \cPp$ since $\unorm{f(B_\xi)}_{\Phi}=\infty$ when $B_\xi$ is not invertible. Hence assume that $B_\xi\in\cPp$. Using the spectral decomposition $B_\xi=\sum_{\lambda\in\mathsf{spec}(B_\xi)} \lambda P_\lambda$, where $\mathsf{spec}(B_\xi)$ is the set of  eigenvalues of $B_\xi$, we then have 
\begin{align} 
f(B_\xi+\eps_m I)&=\sum_{\substack{\lambda\in\mathsf{spec}(B_\xi)\\
\lambda+\eps_m<\delta}}f(\lambda+\eps_m) P_\lambda+\sum_{\substack{\lambda\in\mathsf{spec}(B_\xi)\\
\lambda+\eps_m\geq \delta}}f(\lambda+\eps_m) P_\lambda\\
&\leq\sum_{\substack{\lambda\in\mathsf{spec}(B_\xi)\\
\lambda+\eps_m<\delta}}f(\lambda) P_\lambda+\sum_{\substack{\lambda\in\mathsf{spec}(B_\xi)\\
\lambda+\eps_m\geq \delta}}f(\lambda+\eps_m) P_\lambda\\
&\le f(B_\xi)+\sum_{\substack{\lambda\in\mathsf{spec}(B_\xi)\\
\lambda+\eps_m\geq \delta}}f(\lambda+\eps_m) P_\lambda
\end{align}
The claim then follows by the triangle inequality for $\unorm{\cdot}_\Phi$ and the fact that $f(\lambda+\eps_m)\leq \sup_{\delta\leq x\leq \gamma}f(x)<\infty$ for all $\lambda\in\mathsf{spec}(B_\xi)$ with $\lambda+\eps_m\geq \delta$ and for all $\xi\in\Xi$. The last fact is immediately seen from~\eqref{eq:upperlowerboundbxi} and the continuity of $f$.  

\end{proof}


\begin{proof}[Proof of \eqref{it:fcond}~$\implies$~\eqref{it:loglambda}]
The weak majorization relation
\begin{align}
\prod_{i=1}^k\lambda_i(A)\le\prod_{i=1}^k\exp\int_\Xi\log\lambda_i(B_\xi)\,\di\nu(\xi)\, \ .\label{F-4s}
\end{align}
is obvious from \eqref{it:nondecr}~$\implies$~\eqref{it:loglambdaweak}
in Theorem~\ref{th:mainweak} since condition~\eqref{it:fcond} is stronger than \eqref{it:nondecr}.
It remains to prove that equality holds in \eqref{F-4s} when $k=d$. It suffices to prove that
\begin{align}\label{F-14}
\log(\det A)\ge\int_\Xi\log\bigl(\det B_\xi\bigr)\,\di\nu(\xi) \,.
\end{align}
For this, we may assume that $\int_\Xi\log\bigl(\det B_\xi\bigr)\,\di\nu(\xi)> -\infty$ and so $B_\xi\in\cPp$
for $\nu$-a.e.\ $\xi\in\Xi$. So we may assume that $B_\xi\in\cPp$ for all $\xi\in\Xi$. Moreover, replacing $A$,
$B_\xi$ with $\alpha A$, $\alpha B_\xi$ for some $\alpha>0$, we may assume that $B_\xi\le I$ and so
$\lambda_i(B_\xi)\le1$ for all $\xi\in\Xi$ and $i\in[d]$. For every $p>0$, since
\begin{align}
{\frac{1}{d}}\big\|B_\xi^{-p}\big\|_1\le\lambda_d(B_\xi)^{-p}\le\bigl(\det B_\xi)^{-p} \,,
\end{align}
we find that
\begin{align}\label{F-15}
{\frac{1}{p}}\log\biggl({\frac{1}{d}}\big\|B_\xi^{-p}\big\|_1\biggr)\le-\log\bigl(\det B_\xi\bigr) \,.
\end{align}
Applying inequality \eqref{F-IInonweak} to
$\unorm{\cdot}=\|\cdot\|_1$ and $f(x)=x^{-p}$ for any $p>0$ we get
\begin{align}
\log\|A^{-p}\|_1\le\int_\Xi\log\big\|B_\xi^{-p}\big\|_1\,\di\nu(\xi) \, .
\end{align}
This means that
\begin{align}
{\frac{1}{p}}\log\biggl({\frac{1}{d}}\|A^{-p}\|_1\biggr)\le
\int_\Xi{\frac{1}{p}}\log\biggl({\frac{1}{d}}\big\|B_\xi^{-p}\big\|_1\biggr)\,\di\nu(\xi) \,. \label{F-9}
\end{align}
Similarly to \eqref{F-7} and \eqref{F-8} we find that, as $p\searrow0$,
\begin{align}\label{F-12}
{\frac{1}{p}}\log\biggl({\frac{1}{d}}\|A^{-p}\|_1\biggr)&\searrow-{\frac{1}{d}}\log(\det A) \,, \\
{\frac{1}{p}}\log\biggl({\frac{1}{d}}\big\|B_\xi^{-p}\big\|_1\biggr)&\searrow-{\frac{1}{d}}\log\bigl(\det B_\xi\bigr)
\end{align}
for every $\xi\in\Xi$. Thanks to \eqref{F-15} the Lebesgue convergence theorem yields
\begin{align}
\lim_{p\searrow0}\int_\Xi{\frac{1}{p}}\log\biggl({\frac{1}{d}}\big\|B_\xi^{-p}\big\|_1\biggr)\,\di\nu(\xi)
=-{\frac{1}{d}}\int_\Xi\log\bigl(\det B_\xi\bigr)\,\di\nu(\xi) \,.
\end{align}
Therefore, letting $p\searrow0$ in \eqref{F-9} gives \eqref{F-14}, as desired.
\end{proof}

When $\Xi$ is a one-point set, Proposition~\ref{pr:charweaklog} and Theorem~\ref{th:main} with the above proof of the implication \eqref{it:fcond}~$\implies$~\eqref{it:loglambda} yield a new characterization of log-majorization.
\begin{proposition}
\label{pr:charlog}
For $A,B\in\cP$, 
$\vec{\lambda}(A) \prec_{\log} \vec{\lambda}(B)$ if and only if
\begin{align}
\|A^p\|_{(k)}\le\|B^p\|_{(k)},\qquad p \in \mathbb{R} \setminus \{0\} ,\ \ k \in [d] .
\end{align}
\end{proposition}

It is natural to wonder whether the generalized
log-majorization condition~\eqref{it:firstconditionchain}
of Proposition~\ref{lm:chainofimplications}  is equivalent to
a third condition analogous to~\eqref{it:nondecr} of Theorem~\ref{th:mainweak}. The following theorem shows that this is the case under one additional technical assumption. In the statement,  we use
 the same convention for $\unorm{g(A)}$ as introduced in Theorem~\ref{th:main}.
\begin{theorem}\label{th:mainadditional}
With $\Xi$, $\nu$, $A$ and $B_{\xi}$ {given} as above, consider the
additional statement
\begin{enumerate}[(I)]\setcounter{enumi}{2}
\item for every continuous function $g:(0,\infty)\to[0,\infty)$ such that
$x \mapsto g(e^x)$ is convex on $\bR$, and for every unitarily invariant norm $\unorm{\cdot}$,\label{it:gcond} 
\begin{align}
\unorm{g(A)} \le\int_\Xi \unorm{g(B_\xi)}\, \di\nu(\xi)\,  \,.\label{F-IIIineq}
\end{align}\label{F-III} 
\end{enumerate}
Then  
\eqref{it:loglambda}~$\implies$~\eqref{it:gcond}, and
if $\int_\Xi\big\|B_\xi^{-p}\big\|_1\,d\nu(\xi)<\infty$ for some $p>0$, then
\eqref{it:gcond}~$\implies$~\eqref{it:loglambda}.
\end{theorem} 
\begin{remark}\rm
The integrability assumption is essential in the proof of the implication
\eqref{it:gcond}~$\implies$~\eqref{it:loglambda} with use of test functions $x^{-p}$ for $p>0$. Indeed, it is easy
to provide an example of $\Xi$, $\nu$ and $B_{\xi}$ such that
$-\int_\Xi\log\bigl(\det B_\xi\bigr)\,\di\nu(\xi)<\infty$ but $\int_\Xi\big\|B_\xi^{-p}\big\|_1\,d\nu(\xi)=\infty$ for all
$p>0$. Since there is no good test function other than $x^{-p}$, it seems difficult to remove or relax the
integrability assumption.
\end{remark}


\begin{proof}[Proof of \eqref{it:loglambda}~$\implies$~\eqref{it:gcond}]
Assume first that $A,B_\xi\in\cPp$ and $B_\xi\geq \varepsilon I$ for all $\xi\in\Xi$ with some $\varepsilon>0$. Since~\eqref{it:loglambda}  means that
$\vec{\lambda}(\log A)\prec \int_\Xi \vec{\lambda}(\log B_\xi)\,\di\nu(\xi)$, one can apply~\eqref{it:strconvmaj}~$\implies$~\eqref{it:strfunitaryinvmaj} of Theorem~\ref{th:mainmajsimple} to $\log A$, $\log B_\xi$ and $f(x):=g(e^x)$ for $g$~as in~\eqref{it:gcond}. Inequality~\eqref{F-III} then follows. For general $A, B_\xi\in\cP$ satisfying~\eqref{it:loglambda}, with $0<\varepsilon_m\searrow 0$ 
choose~$\a^{(m)}$ and $A^{(m)}$, $m\geq m_0$, as in the proof of~\eqref{it:loglambda}~$\implies$~\eqref{it:fcond}.  By the first case we have
\begin{align} 
\unorm{g(A^{(m)})}_\Phi\le\int_\Xi\unorm{g(B_\xi+\eps_mI)}_\Phi\,\di\nu(\xi) \,. \label{F-10}
\end{align}
When $g(0^+)<\infty$, letting $m\to\infty$ in \eqref{F-10} gives inequality \eqref{F-IIIineq} immediately. When
$g(0^+)=\infty$, the proof is similar to the last part of the proof \eqref{it:loglambda}~$\implies$~\eqref{it:fcond}
by noting that there is a constant $\beta>0$ such that
$\unorm{g(B_\xi+\eps_mI)}_\Phi\le\unorm{g(B_\xi)}_\Phi+\beta$ for all $\xi\in\Xi$ and $m\ge m_0$.
\end{proof}

\begin{proof}[Proof of \eqref{it:gcond}~$\implies$~\eqref{it:loglambda} under the integrability assumption]
The weak majorization relation
\begin{align}\label{F-11}
\sum_{i=1}^k\log\lambda_i(A)\le\sum_{i=1}^k\int_\Xi\log\lambda_i(B_\xi)\,\di\nu(\xi)\ , \qquad k\in [d]\,
\end{align}
is obvious from \eqref{it:gnondecrcond}~$\implies$~\eqref{it:loglambdaweak} in Theorem~\ref{th:mainweak} since
condition \eqref{it:gcond} is stronger than~\eqref{it:gnondecrcond}. It remains to prove that equality holds in \eqref{F-11} when $k=d$. Here, we use the assumption that
$\int_\Xi\big\|B_\xi^{-p_0}\big\|_1\,\di \nu(\xi)<\infty$ for some $p_0>0$. Inequality~\eqref{F-IIIineq} in~\eqref{F-III} is applied to
$\unorm{\cdot}=\|\cdot\|_1$ and $g(x)=x^{-p}$ for any $p>0$, so that
we have $\|A^{-p}\|_1\le\int_\Xi\big\|B_\xi^{-p}\big\|_1\,\di\nu(\xi)$. Therefore,
\begin{align}
{\frac{1}{p}}\log\biggl({\frac{1}{d}}\|A^{-p}\|_1\biggr)
\le{\frac{1}{p}}\log\int_\Xi{\frac{1}{d}}\big\|B_\xi^{-p}\big\|_1\,\di\nu(\xi)\,.
\end{align}
Thanks to Lemma \ref{lm:main3} separately shown below (and \eqref{F-12} as well), letting $p\searrow0$ in
\eqref{F-11} yields
\begin{align}
-{\frac{1}{d}}\log(\det A)\le-{\frac{1}{d}}\int_\Xi\log\bigl(\det B_\xi\bigr)\,\di\nu(\xi)\,,
\end{align}
which gives the desired equality.

\end{proof}

\begin{lemma}\label{lm:main3}
Let $\Xi$, $\nu$ and $B_{\xi}$ be as above. If $\int_\Xi\big\|B_\xi^{-p_0}\big\|_1\,d\nu(\xi)<\infty$ for some
$p_0>0$, then
\begin{align}
\lim_{p\searrow0}\biggl({\frac{1}{p}}\log\int_\Xi{\frac{1}{d}}\big\|B_\xi^{-p}\big\|_1\,\di\nu(\xi)\biggr)
=-{\frac{1}{d}}\int_\Xi\log\bigl(\det B_\xi\bigr)\,\di\nu(\xi)\,.
\end{align}
\end{lemma}

\begin{proof}
The following proof is similar to that of \cite[Lemma 6.12]{BH14}.
The assumption implies that $B_\xi\in\cPp$ for $\nu$-a.e.\ $\xi\in\Xi$. So we may assume that $B_\xi\in\cPp$ for
all $\xi\in\Xi$. Moreover, replacing $B_\xi$ with $\alpha B_\xi$ for some $\alpha>0$, we may assume that
$B_\xi\le I$. Let $\widetilde\nu:=\nu\otimes\mu$ be the product measure of $\nu$ and the uniform probability
measure $\mu$ on $[d]$. Define
\begin{align}
\phi(\xi,i,p):=\lambda_i(B_\xi)^{-p},\qquad\xi\in\Xi \,,\ \ i\in[d],\ \ p>0.
\end{align}
It is clear that
\begin{align}\label{F-13}
\int_\Xi{\frac{1}{d}}\big\|B_\xi^{-p}\big\|_1\,\di\nu(\xi)
=\int_{\Xi\times[d]}\phi(\xi,i,p)\,\di\widetilde\nu(\xi,i)\,,\qquad p>0.
\end{align}
Hence $(\xi,i)\mapsto\phi(\xi,i,p_0)$ is integrable with respect to $\widetilde\nu$. According to the mean value theorem applied to the function $p\mapsto \lambda_i(B_\xi)^p$, we have 
\begin{align}
{\frac{\phi(\xi,i,p)-\phi(\xi,i,0)}{p}}=-\lambda_i(B_\xi)^{-\theta p}\log\lambda_i(B_\xi)
\le-\lambda_i(B_\xi)^{-p}\log\lambda_i(B_\xi) \,,
\end{align}
for some  $\theta\in(0,1)$ depending on $\xi,i,p$, and
\begin{align}
\lim_{p\searrow0}{\frac{\phi(\xi,i,p)-\phi(\xi,i,0)}{p}}=-\log\lambda_i(B_\xi) \,.
\end{align}
Furthermore, when $0<p<p_1<p_0$, we have
\begin{align}
-\lambda_i(B_\xi)^{-p}\log\lambda_i(B_\xi)&\le-\lambda_i(B_\xi)^{-p_1}\log\lambda_i(B_\xi) \\
&=\lambda_i(B_\xi)^{-p_0}\bigl\{-\lambda_i(B_\xi)^{p_0-p_1}\log\lambda_i(B_\xi)\bigr\} \,.
\end{align}
Since $\sup_{0<\lambda\le1}\bigl(-\lambda^{p_0-p_1}\log\lambda\bigr)<\infty$, we find that
\begin{align*}
(\xi,i)\mapsto-\lambda_i(B_\xi)^{-p_1}\log\lambda_i(B_\xi)
\end{align*} 
is integrable with respect to $\widetilde\nu$.
Hence the Lebesgue convergence theorem yields 
\begin{align}
{\frac{\di}{\di p}}\int_{\Xi\times[d]}\phi(\xi,i,p)\,\di\widetilde\nu(\xi,i)\bigg|_{p=0^+}
&=\lim_{p\searrow0}\int_{\Xi\times[d]}{\frac{\phi(\xi,i,p)-\phi(\xi,i,0)}{p}}\,\di\widetilde\nu(\xi,i) \\
&=-\int_{\Xi\times[d]}\log\lambda_i(B_\xi)\,\di\widetilde\nu(\xi,i) \\
&=-\int_\Xi{\frac{1}{d}}\sum_{i=1}^d\log\lambda_i(B_\xi)\,\di\nu(\xi) \\
&=-{\frac{1}{d}}\int_\Xi\log\bigl(\det B_\xi\bigr)\,\di\nu(\xi) \,,
\end{align}
where ${\frac{\di}{\di p}}(\cdot)\big|_{p=0^+}$ means the right derivative at $p=0$. Now we obtain the desired
equality since
\begin{align}
\lim_{p\searrow0}\biggl({\frac{1}{p}}\log\int_\Xi{\frac{1}{d}}\big\|B_\xi^{-p}\big\|_1\,\di\nu(\xi)\biggr)
&={\frac{\di}{\di p}}\int_{\Xi\times[d]}\phi(\xi,i,p)\,\di\widetilde\nu(\xi,i)\bigg|_{p=0^+}
\end{align}
as easily seen from \eqref{F-13}.
\end{proof}


\section{Application to multivariate norm inequalities}
\label{sec:app}

We recall the inequality~\cite[Thm.~3.2]{sutter16} specialized to the operator norm. For $A_\ell \in \cP$, $\ell \in [n]$ and $\theta \in (0, 1]$, we have
\begin{align}
\label{eq:sbt}
   \log \left\| \left| \prod_{\ell=1}^n A_\ell^\theta \right|^{\frac{1}{\theta}} \right\|
    \leq \int_{-\infty}^{\infty}   \log \left\| \prod_{\ell=1}^n A_\ell^{1+\ci t} \right\|\,\di\beta_{\theta}(t)  \,,
\end{align}
where
\begin{align}
  \di\beta_\theta(t):={\frac{\sin(\pi \theta)}{2{\theta}\bigl(\cos(\pi t)+\cos(\pi \theta)\bigr)}}\,\di t \,,
\end{align}
{and the functional calculus $A_\ell^z$ for any $z\in\mathbb{C}$ is defined with the convention that $0^z=0$.}
A concise proof of this special case is given in Appendix~\ref{app:shorter}. 
Using the rules of antisymmetric tensor power calculus presented in Lemma~\ref{lm:anti}, we find
\begin{align}
  \left| \prod_{\ell=1}^n (\land^{k} A_{\ell} )^\theta \right|^{\frac{1}{\theta}} = \land^{k} \left| \prod_{\ell=1}^n  A_{\ell}^\theta \right|^{\frac{1}{\theta}} 
  \quad \textrm{and} \quad \left| \prod_{\ell=1}^n \left( \land^{k} A_\ell\right)^{1+\ci t} \right| = \land^{k} \left| \prod_{\ell=1}^n A_\ell^{1+\ci t} \right| .
\end{align}
The inequality~\eqref{eq:sbt} applied to the matrices $\land^k A_\ell$ for all $k \in [d]$ thus immediately yields the {log-majorization} relation
\begin{equation}\label{F-6}
\log \vec{\lambda} \Biggl(\bigg|\prod_{\ell=1}^nA_\ell^\theta \bigg|^{\frac{1}{\theta}}\Biggr) \prec
 \int_{-\infty}^\infty \log
\vec{\lambda} \Biggl(\bigg|\prod_{\ell=1}^n A_\ell^{1+it}\bigg|\Biggr)\, \di \beta_\theta(t) \, ,
\end{equation}
where in particular the equality condition for {log-majorization} is satisfied since
\begin{align}
\det\bigg|\prod_{\ell=1}^nA_\ell^\theta\bigg|^{\frac{1}{\theta}}
= \det \bigg| \prod_{\ell=1}^n A_\ell^{1 + \ci t} \bigg|
= \prod_{\ell=1}^n\det A_\ell \,.
\end{align}
Hence we arrive at the following application of Theorem~\ref{th:main}
and Theorem~\ref{th:mainadditional}. Here  we again use the continuous extension and convention of Theorem~\ref{th:main}. %
\begin{corollary}
\label{cor:sbt}
Let $A_\ell \in \cP$ for $\ell \in [n]$, $\theta \in (0, 1]$ and $\unorm{\cdot}$ a unitarily invariant norm. Then, for any continuous function $f:(0,\infty)\to[0,\infty)$ such that $x \mapsto \log f ( {e^x})$ is convex on~$\mathbb{R}$, we have
\begin{align}
\log \Bigg|\!\Bigg|\!\Bigg| f\Biggl(\bigg|\prod_{\ell=1}^nA_\ell^\theta\bigg|^{\frac{1}{\theta}}\Biggr)\Bigg|\!\Bigg|\!\Bigg| \le
\int_{-\infty}^\infty  \log  \Bigg|\!\Bigg|\!\Bigg| f\Biggl(\bigg|\prod_{\ell=1}^n A_\ell^{1+it}\bigg|\Biggr)
\Bigg|\!\Bigg|\!\Bigg|\, \di\beta_\theta(t) \,. \label{eq:cor1}
\end{align}

Moreover, for any continuous  function $g:(0,\infty)\to[0,\infty)$ such that $x \mapsto g({e^x})$ is convex on~$\mathbb{R}$, we have
\begin{align} 
\Bigg|\!\Bigg|\!\Bigg| g\Biggl(\bigg|\prod_{\ell=1}^nA_\ell^\theta\bigg|^{\frac{1}{\theta}}\Biggr) \Bigg|\!\Bigg|\!\Bigg|
\le
\int_{-\infty}^\infty 
\Bigg|\!\Bigg|\!\Bigg| g\Biggl(\bigg|\prod_{\ell=1}^nA_\ell^{1+it}\bigg|\Biggr)
\Bigg|\!\Bigg|\!\Bigg|\, \di\beta_\theta(t) \,. \label{eq:cor2}
\end{align}
\end{corollary}
These inequalities generalize and strengthen the results in~\cite{sutter16}.
For example, consider the function $f: x \mapsto x^{q}$ for $q\in\mathbb{R}\setminus\{0\}$ and the {trace norm} to find
\begin{corollary}
Let $A_\ell\in\cP$ for $\ell\in [n]$ and $\theta \in (0,1]$. Then 
we have
\begin{align}
\log \tr \bigg|\prod_{\ell=1}^nA_\ell^\theta \bigg|^{\frac{q}{\theta}}
 \le
\int_{-\infty}^\infty \log  \tr \bigg|\prod_{\ell=1}^n A_\ell^{1+it}\bigg|^q\,\di\beta_\theta(t) \,, \label{eq:unify}
\end{align}
for any $q\in\mathbb{R}\backslash\{0\}$. 
\end{corollary}
Indeed, this is a strengthening of both~\cite[Thm.~3.2]{sutter16} (which establishes~\eqref{eq:unify} for $q\geq 1$) and~\cite[Thm.~2.3]{sutter16} (which establishes a looser bound for $q > 0$  where the integration on the right-hand side of~\eqref{eq:unify} is replaced by a supremum over $t$) to the case of arbitrary non-zero $q\in\mathbb{R}$.

Finally note that if $A_\ell \in \cPp$, all of these inequalities remain valid in the limit $\theta \to 0$, where the Lie-Trotter product formula asserts that
\begin{align}
  \bigg|\prod_{\ell=1}^nA_\ell^\theta\bigg|^{\frac{1}{\theta}} \longrightarrow \exp \left( \sum_{\ell=1}^n \log A_\ell \right) .
\end{align}
Equations~\eqref{eq:cor1} and~\eqref{eq:cor2} thus hold with this substitution and $\theta = 0$.

\begin{corollary}
\label{cor:sbt-gt}
Let $A_{\ell} \in \cPp$ for $\ell \in [n]$. With $\unorm{\cdot}$, $f$ and $g$ given as in Corollary~\ref{cor:sbt}, \begin{align}
\log \Bigg|\!\Bigg|\!\Bigg| f\Biggl(\exp \left( \sum_{\ell=1}^n \log A_\ell \right)\Biggr)\Bigg|\!\Bigg|\!\Bigg| &\le
\int_{-\infty}^\infty  \log  \Bigg|\!\Bigg|\!\Bigg| f\Biggl(\bigg|\prod_{\ell=1}^n A_\ell^{1+it}\bigg|\Biggr)
\Bigg|\!\Bigg|\!\Bigg|\,\di\beta_0(t) \,, \\
\Bigg|\!\Bigg|\!\Bigg| g\Biggl(\exp \left( \sum_{\ell=1}^n \log A_\ell \right)\Biggr) \Bigg|\!\Bigg|\!\Bigg|
&\le
\int_{-\infty}^\infty 
\Bigg|\!\Bigg|\!\Bigg| g\Biggl(\bigg|\prod_{\ell=1}^nA_\ell^{1+it}\bigg|\Biggr)
\Bigg|\!\Bigg|\!\Bigg|\, \di\beta_0(t) \,. 
\end{align}
\end{corollary}
\noindent These inequalities generalize~\cite[Cor.~3.3]{sutter16}, 
 where the result was shown for the norms $\|\cdot\|_p$ with $p\geq 1$ and
$f$ and~$g$ equal to the identity function. Using this inequality with~$n=4$ and $p=2$,
the authors of~\cite{sutter16} obtained the best currently known lower bound on the remainder term in the strong subadditivity inequality involving the universal rotated Petz recovery map introduced in~\cite{JungeRennerSutterWildeWinter16}. (The first such remainder terms involving recovery maps were recently presented in~\cite{fawzirenner14}.) It remains an open problem whether this application to quantum information can be extended using the  strengthened inequalities obtained here.

\subsection*{Acknowledgements} 
We thank the anonymous referees for their suggestions on this manuscript. FH and MT thank the Zentrum Mathematik at Technische Uni\-ver\-sit\"at M\"un\-chen for its hospitality while part of this work was completed. MT thanks Mario Berta and David Sutter for helpful discussions. {FH acknowledges support by Grant-in-Aid for Scientific Research (C)26400103.} MT is funded by an ARC Discovery Early Career Researcher Award (DECRA) fellowship and acknowledges support from the ARC Centre of Excellence for Engineered Quantum Systems (EQUS). RK is supported by the Technische {Universit\"at M\"unchen} - Institute for Advanced Study, funded by the German Excellence Initiative and the European Union Seventh Framework Programme under grant agreement no. 291763. He acknowledges additional support by DFG project no.~KO5430/1-1. 


\appendix

\section{A short proof of~\eqref{eq:sbt}}
\label{app:shorter}

Hirschman's strengthening of Hadamard's three line theorem~\cite{hirschman52} reads:
\begin{lemma}[Hirschman]
\label{lm:hirschman} Let $S:=\left\{  z\in\mathbb{C}:0\leq\operatorname{Re}(z) \leq1\right\}$ and let $g(z)$ be uniformly bounded on $S$, holomorphic on the interior of $S$ and continuous up to the boundary. Then for $\theta\in[0,1]$, we have
\begin{equation}
\log\left\vert g(\theta)\right\vert \leq\int_{-\infty}^{\infty}\log  \left\vert g(\ci t)\right\vert ^{1-\theta}\,\di\beta_{1-\theta}(t)
+\int_{-\infty}^{\infty}\log  \left\vert g(1+\ci t)\right\vert ^{\theta}
 \,\di\beta_{\theta}(t)\ .
\end{equation}
\end{lemma}

Now let $G(z)$ be a uniformly bounded holomorphic function with values in $\mathbb{C}^{d \times d}$. Fix $\theta \in (0,1)$ and let $u, v \in \mathbb{C}^{d}$ be normalized vectors such that $\langle u, G(\theta) v \rangle = \| G(\theta) \|$. 
Consequently, $g(z) := \langle u, G(z) v \rangle$ can be bounded as $|g(z)| \leq \| G(z) \|$ for all~$z \in S$. It satisfies the assumptions of Hirschman's theorem, yielding
\begin{align}
   \log\| G(\theta) \|
    \leq\int_{-\infty}^{\infty}\log  \| G(\ci t) \|^{1-\theta}\,\di \beta_{1-\theta}(t)
+\int_{-\infty}^{\infty}\log  \| G(1+\ci t) \|^{\theta}\,\di\beta_{\theta}(t)
 \,. \label{eq:stein}
\end{align}

As in~\cite[Thm.~3.2]{sutter16}, consider now a set of $n$ matrices $A_{\ell} \in {\cP}$, $\ell \in [n]$ and set
$G(z) = \prod_{k=1}^n A_{\ell}^z$. Since $G(\ci t)$ is a product of isometries, the first term in the right-hand side of~\eqref{eq:stein} is non-positive and after dividing by $\theta$ we find
\begin{align}
   \log \left\| \left| \prod_{\ell=1}^n A_\ell^{\theta} \right|^{\frac{1}{\theta}} \right\|
    \leq \int_{-\infty}^{\infty} \log  \left\| \prod_{\ell=1}^n A_\ell^{1+\ci t} \right\| \di \beta_{\theta}(t)\,.
\end{align}


\end{document}